\documentclass[journal]{IEEEtran}

\ifCLASSINFOpdf
\else
   \usepackage[dvipdfmx]{graphicx}
\fi
\usepackage{url}
\usepackage{amsmath}
\usepackage{amsthm}
\usepackage{amsfonts}
\usepackage{bm}
\usepackage[subrefformat=parens]{subcaption}
\usepackage{mwe}
\usepackage{breqn}
\usepackage{autobreak}
\usepackage{cite}
\usepackage{multirow}
\usepackage[linesnumbered, titlenumbered, ruled, vlined]{algorithm2e}
\usepackage[table,xcdraw]{xcolor}

\usepackage{booktabs,makecell,tabularx}

\newcolumntype{L}{>{\raggedright\arraybackslash}X}
\usepackage{siunitx}

\newcommand{\argmin}{\mathop{\rm arg~min}\limits}
\newcommand{\st}{{\rm{s.t.}}}
\newcommand{\prox}{{\rm{prox}}}

\newtheorem{theorem}{Proposition}

\begin{document}

\title{Graph Blind Deconvolution with Sparseness Constraint}

\author{Kazuma Iwata, Koki Yamada, \IEEEmembership{Student Member,~IEEE}, and Yuichi Tanaka, \IEEEmembership{Senior Member,~IEEE}
\thanks{K. Iwata and K. Yamada are with the Graduate School of BASE, Tokyo University of Agriculture and Technology, Tokyo 184-8588, Japan (e-mail: k\_iwata@msp-lab.org, k-yamada@msp-lab.org).}
% \thanks{K. Iwata is with the Graduate School of BASE, Tokyo University of Agriculture and Technology, Tokyo 184-8588, Japan (e-mail:k\_iwata@msp-lab.org).}
% \thanks{K. Yamada is with the Graduate School of BASE, Tokyo University of Agriculture and Technology, Tokyo 184-8588, Japan (e-mail:k-yamada@msp-lab.org).}
\thanks{Y. Tanaka is with the Graduate School of BASE, Tokyo University of Agriculture and Technology, Tokyo 184-8588, Japan, and also with the PRESTO, Japan Science and Technology Agency, Kawaguchi 332-0012, Japan (e-mail: ytnk@cc.tuat.ac.jp).}
\thanks{This work was supported in part by JST CREST under grant JPMJCR1784 and JST PRESTO under grant JPMJPR1935.}}

\markboth{}
{}
\maketitle

\begin{abstract}
	We propose a blind deconvolution method for signals on graphs, with the exact sparseness constraint for the original signal. Graph blind deconvolution is an algorithm for estimating the original signal on a graph from a set of blurred and noisy measurements. Imposing a constraint on the number of nonzero elements is desirable for many different applications. This paper deals with the problem with constraints placed on the exact number of original sources, which is given by an optimization problem with an $\ell_0$ norm constraint. We solve this non-convex optimization problem using the ADMM iterative solver. Numerical experiments using synthetic signals demonstrate the effectiveness of the proposed method.
\end{abstract}


\IEEEpeerreviewmaketitle

\section{Introduction}
	\IEEEPARstart{S}{ignals} diffused on a network often have very few original sources. For example, rumors on social networks and spike waves on brain networks begin spreading from very few active sources. An estimation of the source positions on the networks from observed signals, called graph signal deconvolution, is an important task in graph signal processing (GSP) \cite{shuman+,Ortega2018}. GSP is an extension of classical signal processing theory to signals on graphs \cite{shuman+,Ortega2018,Cheung2018,tanaka2020sampling,HAMMOND2011129,Sakiya2014a, Sakiya2019, Sakiya2019a, 6409473,6808520,7926424, Narang2012, Narang2013, Onuki2016, Anis2016, Tanaka2018, Tanaka2020b}. The techniques of graph signal deconvolution are based on classical blind deconvolution algorithms for images \cite{489268,6680763}. Blind deconvolution is a method for restoring an original signal from blurred measurement(s) without knowledge of the information of the spreading, that is, filters. \par
	Extensions of blind deconvolution to the graph domain have been studied in \cite{blind_identification,7952928}. The target signal in these studies was modeled as a signal diffused by a graph filter. Graph filters are a special class of linear operators whose input and output are graph signals. Mathematically, graph filters are defined as a linear transformation that can be expressed as a polynomial of the graph variation operator\cite{shuman+, Ortega2018, 6808520}. Graph blind deconvolution simultaneously estimates the coefficients of the graph filter and the original signal. The number of original sources is expressed using the $\ell_0$ constraint. However, because the $\ell_0$ pseudo-norm is a non-convex function, it is generally difficult to use for optimization. In \cite{blind_identification}, the $\ell_1$-norm constraint, a convex relaxation of the $\ell_0$ pseudo-norm, is used instead of the $\ell_0$ constraint. For this reason, it is not possible to strictly limit the number of signal sources in the restored signal. \par
	In this letter, we consider a graph blind deconvolution problem that estimates an original signal having only a small number of nonzero elements from noisy signals diffused on a graph. In particular, we assume that the number of signal sources $S$ is given a priori. We formulate a non-convex optimization problem with an $S$-sparse constraint and solve it using ADMM \cite{l0_grad}. Non-convex optimization problems often converge to the local minima; however, the iterative ADMM solver works well for such a non-convex optimization with an appropriate initial value \cite{MAL-016}. Our proposed method shares a similar proposal with graph blind deconvolution \cite{blind_identification}; however, the previous study did not specify the number of signal sources of the original signal. Finally, we provide an illustrative experiment conducted on a synthetic dataset and compare the results with those of the conventional method. The results show that our constraint effectively estimates the signal sources, even under noisy situations.
\section{Problem formulation} % (fold)
\label{sec:rerated_work}
	Let $\mathcal{G}=(\mathcal{V},\mathcal{E},\bm{W})$ denote an undirected graph, where $\mathcal{V}$ and $\mathcal{E}$ represent sets of nodes and edges, respectively. An $N\times N$ matrix $\bm{W}$ contains edge weights, with $w_{ij} = w_{ji}$ denoting a positive weight of an edge connecting nodes $i$ and $j$, and $w_{ij} = 0$ if there is no edge. A graph signal defined on $\mathcal{V}$ can be represented as a vector $\bm{x} = [x_0,\dots,x_{N-1}]^T$, where $x_i$ represents the signal value at node $i$. A graph variation operator $\bm{S}$ is a matrix derived from $\bm{W}$. Its examples are graph Laplacian or adjacency matrix. Assuming that $\bm{S}$ is diagonalizable, the graph variation operator can be decomposed into $\bm{S}=\bm{V}\bm{\Lambda}\bm{V}^{-1}$, where $\bm{\Lambda}\in\mathbb{R}^{N\times N}$ is a diagonal matrix. Based on the graph variation operator $\bm{S}$, a linear graph filter is given by
	\begin{align}
		\bm{H}:=\sum_{l=0}^{L-1}h_l\bm{S}^l,
		\label{Eq:graph_filter}
	\end{align}
	where $\bm{h} = [h_0,\dots,h_{L-1}]^T$ represents the filter coefficients. Using the spectral decomposition of $\bm{S}$, the graph filter and signal can be represented in the graph frequency domain. The filtering operation is given by $\bm{y} = \bm{Hx}$, where $\bm{y}$ is the filtered signal and $\bm{x}$ is the original signal.
	\par
	Graph filters and signals can be represented in the frequency domain. Let us define the matrices $\bm{U} = \bm{V}^{-1}\in\mathbb{R}^{N\times N}$ and $\bm{\Psi}\in\mathbb{R}^{N\times L}$, where $\Psi_{ij} = (\Lambda_{ii})^{j-1}$. Using them, the frequency representation of the signal $\bm{x}$ and filter $\bm{h}$ is defined as $\hat{\bm{x}} = \bm{Ux}$ and $\hat{\bm{h}} = \bm{\Psi h}$, respectively. Therefore, given a measurement $\bm{y}\in\mathbb{R}^N$, we can obtain its frequency-domain representation by $\hat{\bm{y}} = \bm{Uy}={\rm{diag}}(\bm{\Psi h})\bm{Ux}$. \par
	Suppose that the number of active sources in the original signal is equal to or less than $S$. The graph blind deconvolution is formulated as the following problem: 
	\begin{align}
		{\rm{find}}\ \ \{\bm{h},\bm{x}\}\ \ \st\ \ \hat{\bm{y}}={\rm{diag}}(\bm{\Psi h})\bm{Ux},\  \|\bm{x}\|_0\leq S.
		\label{Eq:Graph_blind_deconvolution_1}
	\end{align}
	The first constraint in (\ref{Eq:Graph_blind_deconvolution_1}) can be rewritten as $\hat{\bm{y}} = \left(\bm{\Psi}^T\odot\bm{U}^T\right)^T{\rm{vec}}(\bm{xh}^T)$, where $\odot$ denotes the Khatri-Rao product and ${\rm{vec}}(\cdot)$ is the vectorization operator. Let us define $\bm{Z}: = \bm{xh}^T$ and $\bm{M}: = \left(\bm{\Psi}^T\odot\bm{U}^T\right)^T$. The following problem may then be considered from (\ref{Eq:Graph_blind_deconvolution_1}):
	\begin{align}
		\min_{\bm{Z}}\ \ {\rm{rank}}(\bm{Z})\ \ \st\ \ \hat{\bm{y}}=\bm{M}{\rm{vec}}(\bm{Z}),\ \|\bm{Z}\|_{2,0}\leq S,
		\label{Eq:Graph_blind_deconvolution_2}
	\end{align}
	where $\|\bm{Z}\|_{2,0}$ is equal to the number of nonzero rows of $\bm{Z}$.
	The rank and $\ell_{2,0}$ pseudo-norm minimization are generally combinatorial and NP-hard. In \cite{blind_identification}, to make (\ref{Eq:Graph_blind_deconvolution_2}) tractable, the nuclear norm $\|\bm{Z}\|_*$ is utilized as a convex relaxation of the rank function. Similarly, the $\ell_{2,1}$ mixed norm $\|\bm{Z}\|_{2,1}$ is the closest convex relaxation of $\|\bm{Z}\|_{2,0}$ \cite{1603770}. As a result, in the existing method \cite{blind_identification}, problem (\ref{Eq:Graph_blind_deconvolution_2}) is transformed into the following convex optimization problem:
	\begin{align}
		\min_{\bm{Z}}\ \ \|\bm{Z}\|_*+\tau\|\bm{Z}\|_{2,1}\ \ \st\ \ \hat{\bm{y}}=\bm{M}{\rm{vec}}(\bm{Z}).
		\label{Eq:Graph_blind_deconvolution_3}
	\end{align}\par
	The accuracy of the estimation can be improved by using multiple measurements. 
	We consider $P$ measurements $\{\bm{y}_p\}_{p=1}^P$, where each different sparse input is diffused by the common filter $\bm{H}$. Multiple signals are then treated as a vector of stacked measurements $\tilde{\bm{y}}=[\hat{\bm{y}}_1^T,\dots,\hat{\bm{y}}_P^T]^T\in\mathbb{R}^{NP}$, and similarly for the unobserved inputs $\tilde{\bm{x}}=[\bm{x}_1^T,\dots,\bm{x}_P^T]^T\in\mathbb{R}^{NP}$.
	In addition, the matrices are $\bm{Z}_p=\bm{x}_p\bm{h}^T, p=1,\dots,P$, and let the vertical and horizontal matrices be $\tilde{\bm{Z}}_v=[\bm{Z}_1^T,\dots,\bm{Z}_P^T]^T\in\mathbb{R}^{NP\times L}$ and $\tilde{\bm{Z}}_h=[\bm{Z}_1,\dots,\bm{Z}_P]\in\mathbb{R}^{N\times PL}$, respectively. The formulation using multiple measurements is then given as follows: 
	\begin{align}
		\begin{split}
			\min_{\{\bm{Z}_p\}_{p=1}^P}\ \ \|\tilde{\bm{Z}}_v\|_*+\tau\|\tilde{\bm{Z}}_h\|_{2,1}\ \ \st\ \ \tilde{\bm{y}}=(\bm{I}_P\otimes\bm{M}){\rm{vec}}(\tilde{\bm{Z}}_h).
		\end{split}
		\label{Eq:multiple_input}
	\end{align}
% section rerated_work (end)
\section{$S$-sparse constraint} % (fold)
\label{sec:S_sparse_constraint}
	\subsection{Formulation of Sparseness Constraint} % (fold)
	\label{sub:}
		For simplicity, we consider the single-input case in (\ref{Eq:Graph_blind_deconvolution_3}). However, the method introduced in this section can be easily extended to the multiple-input case in (\ref{Eq:multiple_input}).\par
		In fact, problem (\ref{Eq:Graph_blind_deconvolution_3}) cannot strictly limit the number of signal sources, although it is a convex optimization problem. Instead, we consider the following problem to constrain the exact sparseness of $\bm{Z}$.
		\begin{align}
			\min_{\bm{Z}}\ {\rm{rank}}(\bm{Z})\ \st \ \|\hat{\bm{y}}-\bm{M}{\rm{vec}}(\bm{Z})\|_{2}<\epsilon, \|\bm{Z}\|_{2,0}\leq S,
			\label{Eq:l0_blind_deconvolution}
		\end{align}
		where $\epsilon>0$. Here, let us define the indicator function of the inequality constraint on the $\ell_{2,0}$ mixed pseudo-norm in (\ref{Eq:l0_blind_deconvolution}) as
		\begin{align}
			\mathcal{I}_{\|\cdot\|_{2,0}\leq S}(\bm{X})=\begin{cases}
				0 & {\rm{if}}\ \|\bm{X}\|_{2,0}\leq S,\\
				\infty & {\rm{otherwise}}.
			\end{cases}
			\label{Eq:l0_indicator}
		\end{align}
		The graph blind deconvolution with the $S$-sparse constraint in (\ref{Eq:l0_blind_deconvolution}) can then be reformulated as follows: 
		\begin{align}
			\min_{\bm{Z}}\ \|\bm{Z}\|_*+\mathcal{I}_{\|\cdot\|_{2,0}\leq S}(\bm{Z})\ \ \st \ \|\hat{\bm{y}}-\bm{M}{\rm{vec}}(\bm{Z})\|_{2}<\epsilon,
			\label{Eq:l0_blind_deconvolution2}
		\end{align}
		where we use the relaxation of the rank function as in (\ref{Eq:Graph_blind_deconvolution_3}). Further, we modify (\ref{Eq:l0_blind_deconvolution2}) by introducing local variables $\bm{Z}_1,\bm{Z}_2\in\mathbb{R}^{N\times L}$ such that the problem can be applied to an iterative solver based on the ADMM:
		\begin{align}
			\begin{split}
				&\min_{\bm{Z}}\ \|\bm{Z}_1\|_*+\mathcal{I}_{\|\cdot\|_{2,0}\leq S}(\bm{Z}_2)+\mathcal{I}_{\mathcal{D}}(\bm{W})\\
				&\ \st\ \ \bm{Z}_i=\bm{W},\ i=1,2,
			\end{split}
			\label{Eq:l0_admm}
		\end{align}
		where $\mathcal{I}_{\mathcal{D}}(\cdot)$ is an indicator function for $\mathcal{D}=\{\bm{Z}\in\mathbb{R}^{N\times L}\ |\ \|\hat{\bm{y}}-\bm{M}{\rm{vec}}(\bm{Z})\|_{2}<\epsilon\}$. The constraint in (\ref{Eq:l0_admm}) ensures that all local variables are identical to $\bm{W}$. We solve (\ref{Eq:l0_admm}) with the following iterations:
		\begin{align}
			\begin{split}
				\bm{W}^{(n+1)}=&\argmin_{\bm{W}} \ \ \mathcal{I}_{\mathcal{D}}(\bm{W})\\
				&\ \ \ \ \ \ \ \ \ \ \ +\rho\|\bm{W}-\frac{1}{2}\sum_{k=1}^2(\bm{Z}_k^{(n)}-\bm{Y}_k^{(n)})\|_F,
			\end{split}\label{Eq:admm_step1}\\
			\bm{Z}_1^{(n+1)}=&\argmin_{\bm{Z}}\ \ \|\bm{Z}\|_*+\frac{\rho}{2}\|\bm{Z}-(\bm{W}^{(n+1)}+\bm{Y}_1^{(n)})\|_F,\\
			\begin{split}
				\bm{Z}_2^{(n+1)}=&\argmin_{\bm{Z}}\ \ \mathcal{I}_{\|\cdot\|_{2,0}\leq S}(\bm{Z})\\
				&\ \ \ \ \ \ \ \ \ \ \ +\frac{\rho}{2}\|\bm{Z}-(\bm{W}^{(n+1)}+\bm{Y}_2^{(n)})\|_F,
			\end{split}\label{Eq:l0_admm_2}\\
			\bm{Y}_k^{(n+1)}=&\bm{Y}_k^{(n)}+\bm{W}^{(n+1)}-\bm{Z}_k^{(n+1)}\ k=1,2.
			\label{Eq:admm_step3}
		\end{align}
		The indicator function in (\ref{Eq:l0_admm_2}) is non-convex because a set satisfying the $S$-sparse constraint is a non-convex set. Therefore, the optimization problem becomes non-convex. Although ADMM is a method for solving a class of convex optimization problems, it has been validated as effective for non-convex optimization problems in practice \cite{l0_grad}, \cite{BERRY2007155}.
	% subsection subsection_name (end)
	\subsection{Optimization with $S$-sparse Constraint} % (fold)
	\label{sub:projection_onto_mixed_l_}
		The following is equivalent to (\ref{Eq:l0_admm_2}):
		\begin{align}
			{\rm{find}}\ \ \bm{Z}^*\ \in\argmin_{\bm{Z}\in\{\bm{Z}|\|\bm{Z}\|_{2,0}\leq S\}}\|\bm{Z}-\overline{\bm{Z}}\|_F,
			\label{Eq:proj_l20}
		\end{align}
		where $\overline{\bm{Z}}=\bm{W}^{(n+1)}+\bm{Y}_2^{(n)}$. This is the projection onto a set satisfying the $S$-sparse constraint. In other words, minimization can be performed by calculating the projection onto the $\ell_{2,0}$ mixed pseudo-norm ball. Projection (\ref{Eq:proj_l20}) might appear to be difficult; however, its optimal solution can be computed in a closed form, which is given by the following result:
		\begin{theorem} 
			Let $\overline{\bm{Z}}=[\overline{\bm{z}}_1,\dots,\overline{\bm{z}}_N]^T$, i.e., $\overline{\bm{z}}_1,\dots,\overline{\bm{z}}_N$ are the rows of $\overline{\bm{Z}}$ in (\ref{Eq:proj_l20}). In addition, let $\overline{\bm{z}}_{(1)},\dots,\overline{\bm{z}}_{(N)}$ be the vectors $\overline{\bm{z}}_1,\dots,\overline{\bm{z}}_N$ sorted in descending order in terms of their $\ell_2$ norms, that is, $\|\overline{\bm{z}}_{(1)}\|_2\geq\|\overline{\bm{z}}_{(2)}\|_2\geq\cdots\geq\|\overline{\bm{z}}_{(N)}\|_2$. The index $\cdot_{(k)}$ corresponds to the index of the $k$-th largest row in terms of their norm. The projection {(\ref{Eq:proj_l20})} can be written as follows: 
			\begin{align}
				{\rm{find}}\ \ \bm{Z}^*\ \in\argmin_{\bm{Z}\in\mathbb{R}^{N\times L}}\ \ \|\bm{Z}-\overline{\bm{Z}}\|_F^2\ \ \st\ \ \|\bm{Z}\|_{2,0}\leq S.
				\label{Eq:nonconvex_optimization2}
			\end{align}
			One minima of (\ref{Eq:nonconvex_optimization2}) is given by
			\begin{align}
				\bm{Z}^*=\begin{cases}
					\overline{\bm{Z}} & {\rm{if}}\ \|\overline{\bm{Z}}\|_{2,0}\leq S,\\
					[\tilde{\bm{z}}_1^T,\dots,\tilde{\bm{z}}_N^T]^T & {\rm{if}}\ \|\overline{\bm{Z}}\|_{2,0}> S,
				\end{cases}
			\end{align}
			where
			\begin{align}
				\tilde{\bm{z}}_i=\begin{cases}
					\overline{\bm{z}}_{(i)} & {\rm{if}}\ i\in\{1,\dots,S\},\\
					\bm{0} & {\rm{if}}\ i\in\{S+1,\dots,N\}.
				\end{cases}
			\end{align}
			Thus, the projection onto the $\ell_{2,0}$ mixed pseudo-norm ball is equivalent to preserving the top $S$ rows of $\overline{\bm{Z}}$ according to their $\ell_2$ norms.
		\end{theorem}
		\begin{proof}
			Because the case of $\|\overline{\bm{Z}}\|_{2,0}\leq S$ is trivial, we consider the case of $\|\overline{\bm{Z}}\|_{2,0}> S$. To satisfy the inequality constraint $\|\bm{Z}\|_{2,0}\leq S$ in (\ref{Eq:nonconvex_optimization2}), at least $N-S$ subvectors of $\bm{Z}^*$ must be zero vectors from the definition of $\ell_{2,0}$ mixed pseudo-norm. Meanwhile, any change in $\bm{Z}^*$ from $\overline{\bm{Z}}$ increases the value of $\|\bm{Z}-\overline{\bm{Z}}\|_F^2$. From these facts, the $k$-th subvector of $\bm{Z}^*$ of the optimal solution must consist of $\overline{\bm{z}}_{k}$ or $\bm{0}$. Therefore, the cost function is expressed as
			\begin{align}
				\|\bm{Z}-\overline{\bm{Z}}\|_F^2=\sum_{k=1}^N\|\bm{z}_k-\overline{\bm{z}}_{k}\|_2^2.
				\label{Eq:cost_function}
			\end{align}
			If we set $\bm{z}_k = \bm{0}$, the cost is increased by $\|\overline{\bm{z}}_{k}\|_2^2$. Hence, from the fact that $\|\overline{\bm{z}}_{(1)}\|_2\geq\|\overline{\bm{z}}_{(2)}\|_2\geq\cdots\geq\|\overline{\bm{z}}_{(N)}\|_2$, we can conclude that setting $\bm{z}_{(1)}=\overline{\bm{z}}_{(1)},\dots,\bm{z}_{(S)}=\overline{\bm{z}}_{(S)}$ and $\bm{z}_{(S+1)}=\bm{0},\dots,\bm{z}_{(N)}=\bm{0}$ minimizes the cost function (\ref{Eq:cost_function}) subject to the inequality constraint $\|\bm{Z}\|_{2,0}\leq S$.
		\end{proof}
		\par
		Finally, the detailed steps of our algorithm are summarized in Algorithm \ref{Alg:proposed}. In the algorithm, a scalar $\eta$ is set to gradually decrease the value of $\gamma$, which stabilizes the ADMM for non-convex optimization, which is shown in the convergence analysis of the ADMM in the non-convex case \cite{doi:10.1137/140990309}, where the iterations generated by the ADMM under appropriate conditions converge to a stationary point with a sufficiently small $\gamma$. In addition, the solution of the non-convex optimization problem strongly depends on the initial value. Among the computable solutions, the closest to the optimal solution of the non-convex optimization problem is the solution of the convex relaxed optimization problem. Therefore, we recommend using the solution of the convex relaxed problem (\ref{Eq:Graph_blind_deconvolution_3}) as the initial value.
		\begin{algorithm}
			\SetKwInOut{Input}{Input}
			\SetKwInOut{Output}{Output}
			\Input{$\{\hat{\bm{y}}_p\}_{p=1}^P$, $S$, $\bm{Y}_k^{(0)}$,\\
			$\bm{Z}_k^{(0)}$: Arbitrary initial value,\\
			$\rho > 0$, and $0 < \eta < 1$}
			\Output{$\bm{W}^{(n)}$}
			\While{a stopping criterion is not satisfied}{
				$\bm{W}^{(n+1)}=\prox_{\mathcal{I}_{\mathcal{D}}(\cdot)}\left(\frac{1}{2}\sum_{k=1}^2\left(\bm{Z}_k^{(n)}-\bm{Y}_k^{(n)}\right)\right)$\\
				$\bm{Z}_1^{(n+1)}=\prox_{\rho\|\cdot\|_*}\left(\bm{W}^{(n+1)}+\bm{Y}_1^{(n+1)}\right)$\\
				Set $\overline{\bm{Z}}=\bm{W}^{(n+1)}+\bm{Y}_2^{(n+1)}=[\overline{\bm{z}}_1^T,\dots,\overline{\bm{z}}_N^T]^T$\\.
				Compute indices $(1),\dots,(N)$ by sorting $\overline{\bm{z}}_1,\dots,\overline{\bm{z}}_N$ in descending order in terms of their $\ell_2$ norms\\.
				Set $\bm{z}_{(1)}^*=\overline{\bm{z}}_{(1)},\dots,\bm{z}_{(S)}^*=\overline{\bm{z}}_{(S)}$ and $\bm{z}_{(S+1)}^*=\bm{0},\dots,\bm{z}_{(N)}^*=\bm{0}$\\.
				$\bm{Z}_2^{(n+1)}=[{\bm{z}_{1}^*}^T,\dots,{\bm{z}_{N}^*}^T]^T$\\
				\For{$k=1,2$}{
					$\bm{Y}_k^{(n+1)}=\bm{Y}_k^{(n)}+\bm{W}^{(n+1)}-\bm{Z}_k^{(n+1)}$
				}
				$\rho\leftarrow\eta\rho$\\
				$n\leftarrow n+1$
			}
			\caption{Graph Blind Deconvolution with Sparseness Constraint based on ADMM}
			\label{Alg:proposed}
		\end{algorithm}
	% subsection projection_onto_mixed_l_ (end)
% section section_name (end)
\section{Experimental Results} % (fold)
\label{sec:experimental_result}
  	\subsection{Graph filter based diffusion}
  	\label{sub:experiment-A}
  		We validate the performance of the blind deconvolution with the $S$-sparse constraint by solving (\ref{Eq:l0_blind_deconvolution2}) and comparing the result with the recovery result of (\ref{Eq:Graph_blind_deconvolution_3}). We use an undirected random sensor graph with $N = 64$ and a community graph with $N = 100$ \cite{perraudin2014gspbox}. The graph variation operator used is the adjacency matrix of $\mathcal{G}$, that is, $\bm{S} = \bm{A}$. Let $\bm{x}_0$ be the original signal and $\tilde{\bm{x}}$ be the restored signal. The root-mean-square error RMSE$ = \|\tilde{\bm{x}}-\bm{x}_0\|$ is used as an objective measure of the restoration performance.\par
  		Synthetic signals are modeled by $\bm{y} = \bm{Hx}+\bm{\epsilon}$, where $\bm{\epsilon}$ is an additive white Gaussian noise and the filter coefficients with $L = 3$ are set to $\bm{h} = [1.0, 0.8, 0.3]$. The number of nonzero elements in the original signal is $S = 3$. Figs. \ref{Fig:graphs}\subref{Fig:original} and \ref{Fig:community_graphs}\subref{Fig:community_original} show the signal sources and Figs. \ref{Fig:graphs}\subref{Fig:diffusion} and \ref{Fig:community_graphs}\subref{Fig:community_diffusion} are examples of the diffused noisy measurements. We generated 30 synthetic signals with random source locations to conduct the restoration experiment.\par
  		\begin{figure}
			\centering
			\begin{subfigure}[b]{0.45\hsize}
				\centering
				\includegraphics[width=\hsize]{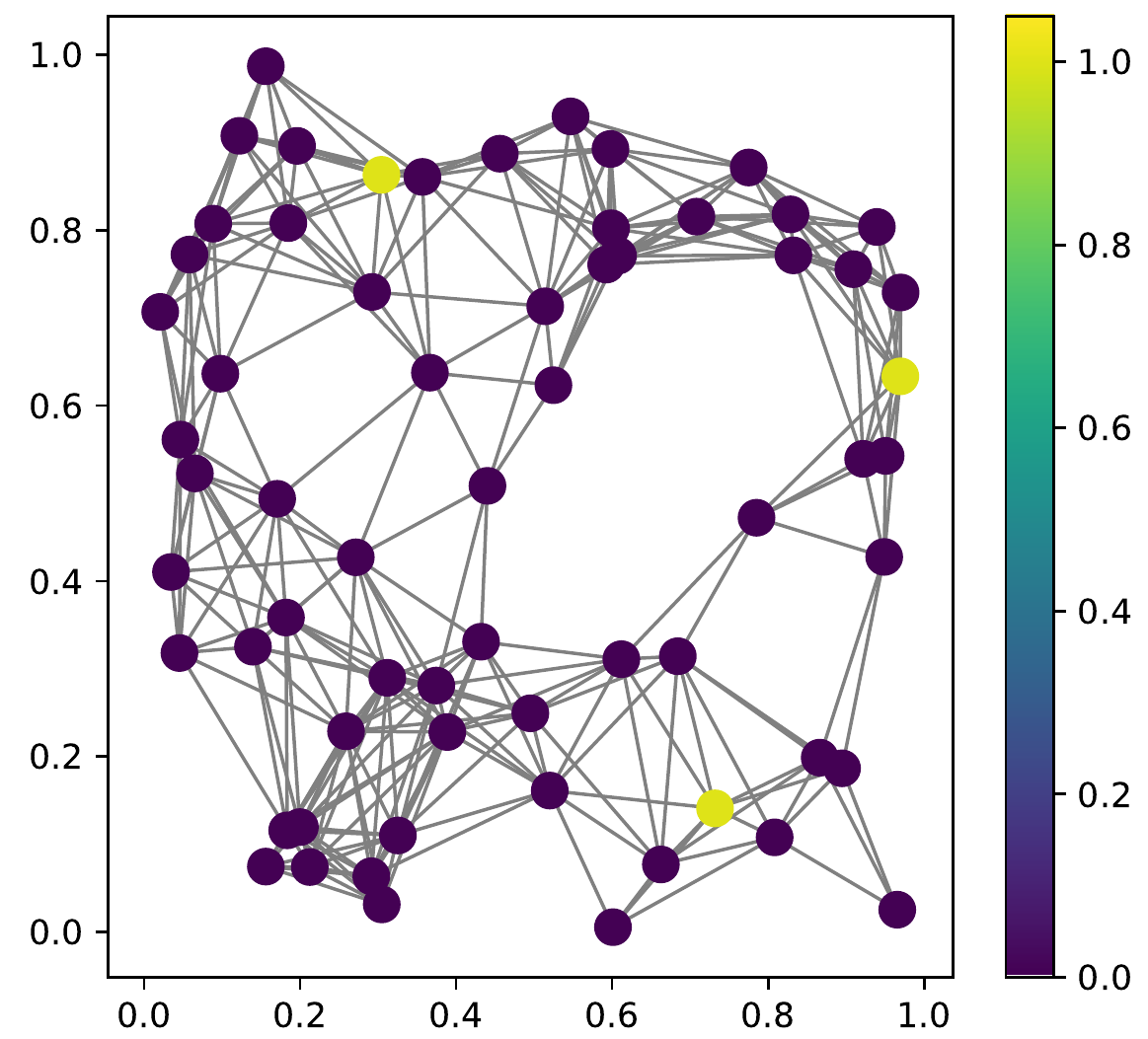}
				\caption{original signal.\protect\linebreak}
				\label{Fig:original}
			\end{subfigure}
			\hfill
%             \hspace{-80pt}
			\begin{subfigure}[b]{0.45\hsize}
				\centering
				\includegraphics[width=\hsize]{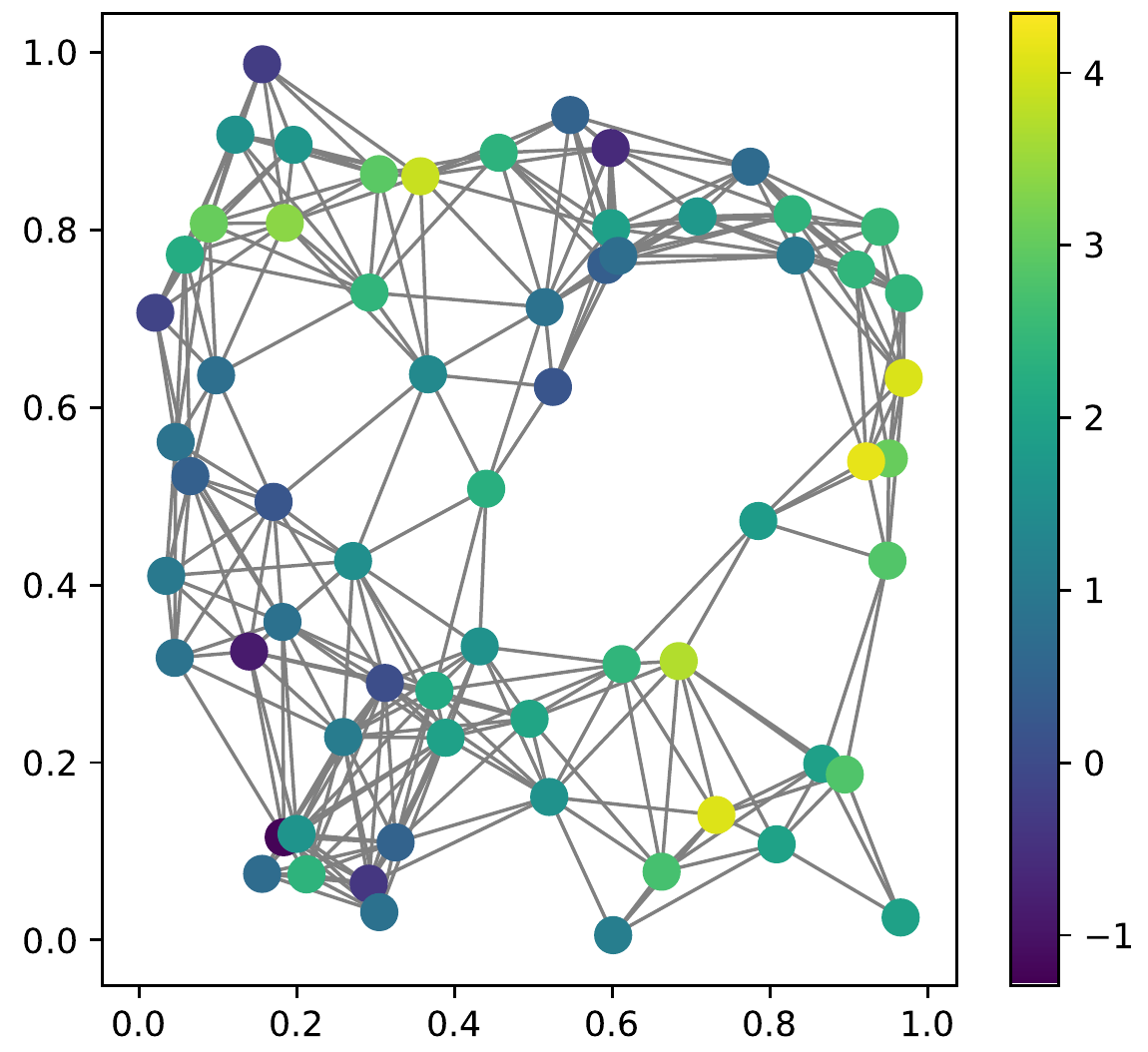}
				\caption{diffused signal (noisy).\protect\linebreak\centering RMSE=$4.1$}
				\label{Fig:diffusion}
			\end{subfigure}
			\vskip\baselineskip
%             \vspace{-10pt}
			\begin{subfigure}[b]{0.45\hsize}
				\centering
				\includegraphics[width=\hsize]{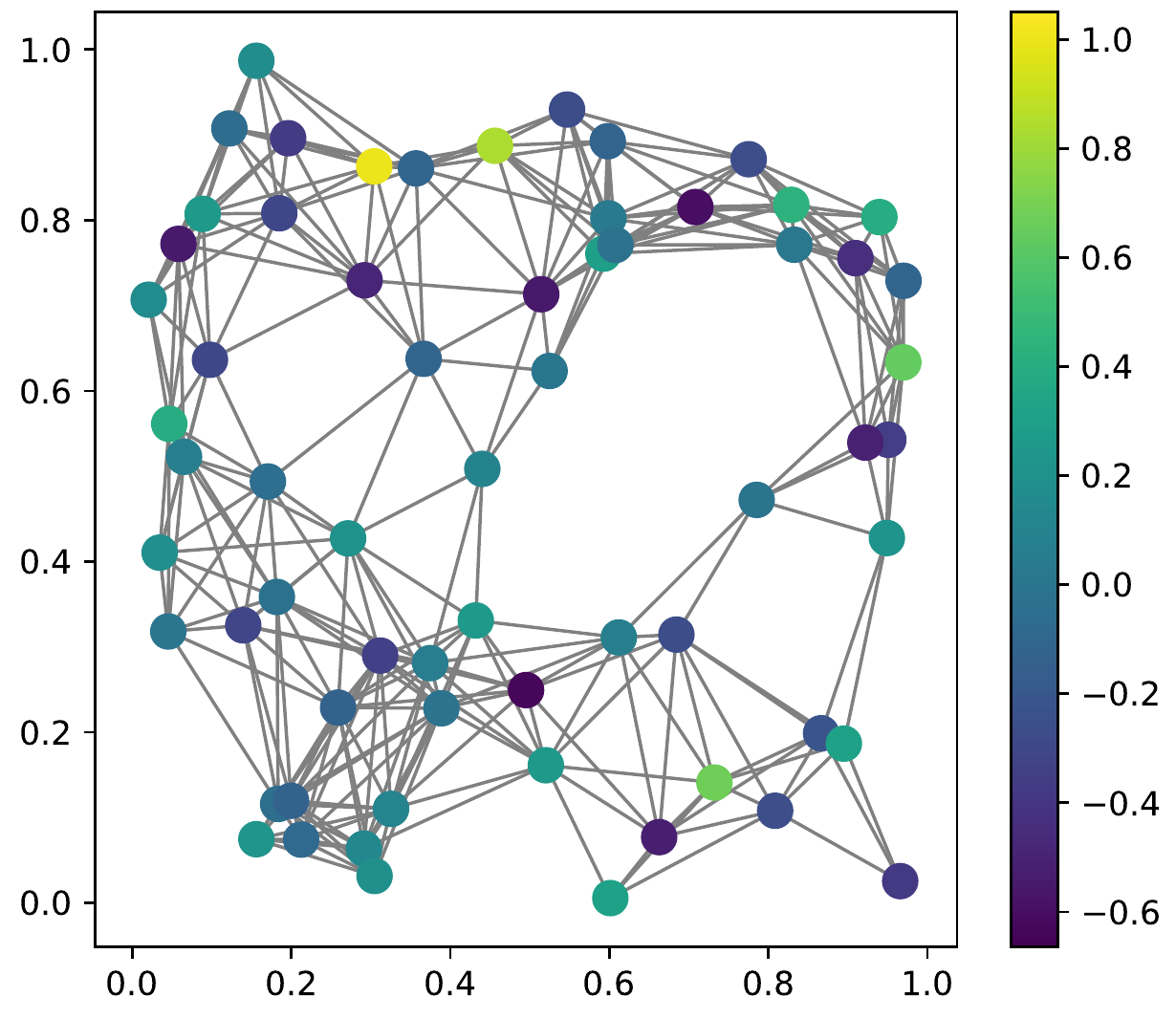}
				\caption{convex relaxation \cite{blind_identification}.\protect\linebreak\centering RMSE=$3.0\times 10^{-1}$}
				\label{Fig:convex}
			\end{subfigure}
			\hfill
% 			\hspace{-80pt}
			\begin{subfigure}[b]{0.45\hsize}
				\centering
				\includegraphics[width=\hsize]{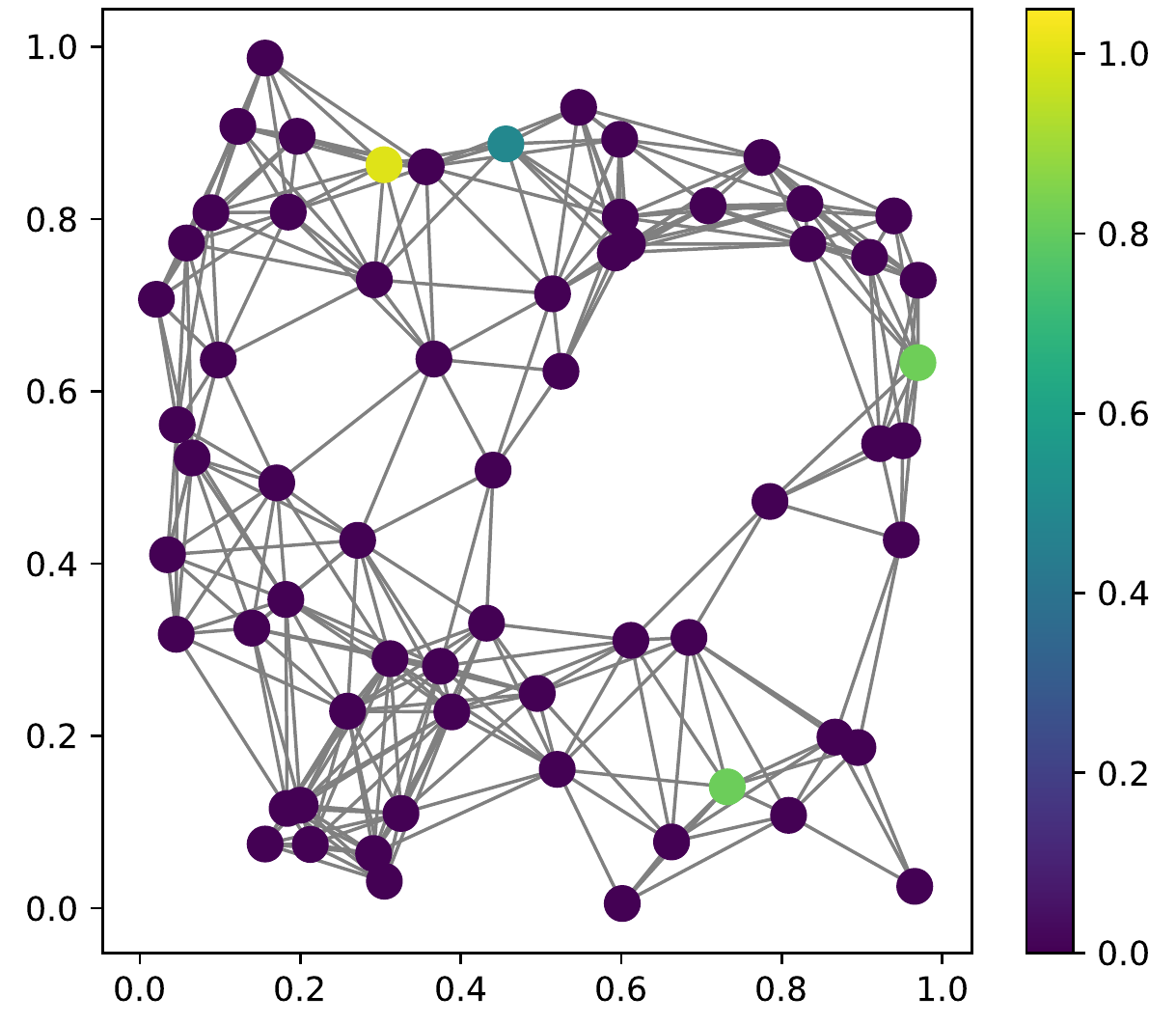}
				\caption{$S$-sparse constraint.\protect\linebreak\centering RMSE=$6.9\times 10^{-2}$}
				\label{Fig:ssparse}
			\end{subfigure}
			\caption{Recovery results on a sensor graph with $N = 64$ and $L = 3$.}
			\label{Fig:graphs}
		\end{figure}
		\begin{table}
			\caption{Average RMSE of restored signals in Experiment \ref{sub:experiment-A}}
			\label{Tab:rmse-A}
			\centering
			\setlength{\tabcolsep}{5pt}
	   % 	\small
			\begin{tabularx}{\hsize}{c|ccc}
				\hline
				\thead{Graph}    & \thead{Diffused signal} & \thead{Convex relaxation} & \thead{Proposed} \\\hline
				% \midrule
				Sensor   & $3.8$ & $3.5\times 10^{-1}$  & $\bm{5.8\times 10^{-2}}$ \\
		        Community & $1.5\times 10^{1}$ & $4.3\times 10^{-1}$  & $\bm{1.4\times 10^{-2}}$ \\
				\hline
			\end{tabularx}	
		\end{table}
		In all recovery experiments, we used an optimal solution of (\ref{Eq:Graph_blind_deconvolution_3}) as the initial value $\bm{Z}_k^{(0)}$ of the minimization problem in (\ref{Eq:admm_step1})--(\ref{Eq:admm_step3}). The RMSE values for the restoration results are shown in Table \ref{Tab:rmse-A}. \par
		Fig. \ref{Fig:graphs}\subref{Fig:convex} shows the result restored by the convex relaxation (\ref{Eq:Graph_blind_deconvolution_3}). The restored signal values have high magnitudes at the sources of the original signal; however, a few samples other than the original sources also have a high magnitude. Therefore, it is difficult to accurately estimate the position of the signal source from the restored signal. Fig. \ref{Fig:graphs}\subref{Fig:ssparse} shows the restored signal by our $S$-sparse constraint. It can be seen that the position of the signal source of the restored signal is clearly the same as those of the original signal. \par
	\subsection{Estimating the case of a mismatched graph filter order} % (fold)
	\label{sub:experiment-B}	
		% subsection estimating_the_case_of_mismatched_graph_filter_order (end)
		\begin{figure}
			\centering
			\begin{subfigure}[b]{0.45\hsize}
				\centering
				\includegraphics[width=\hsize]{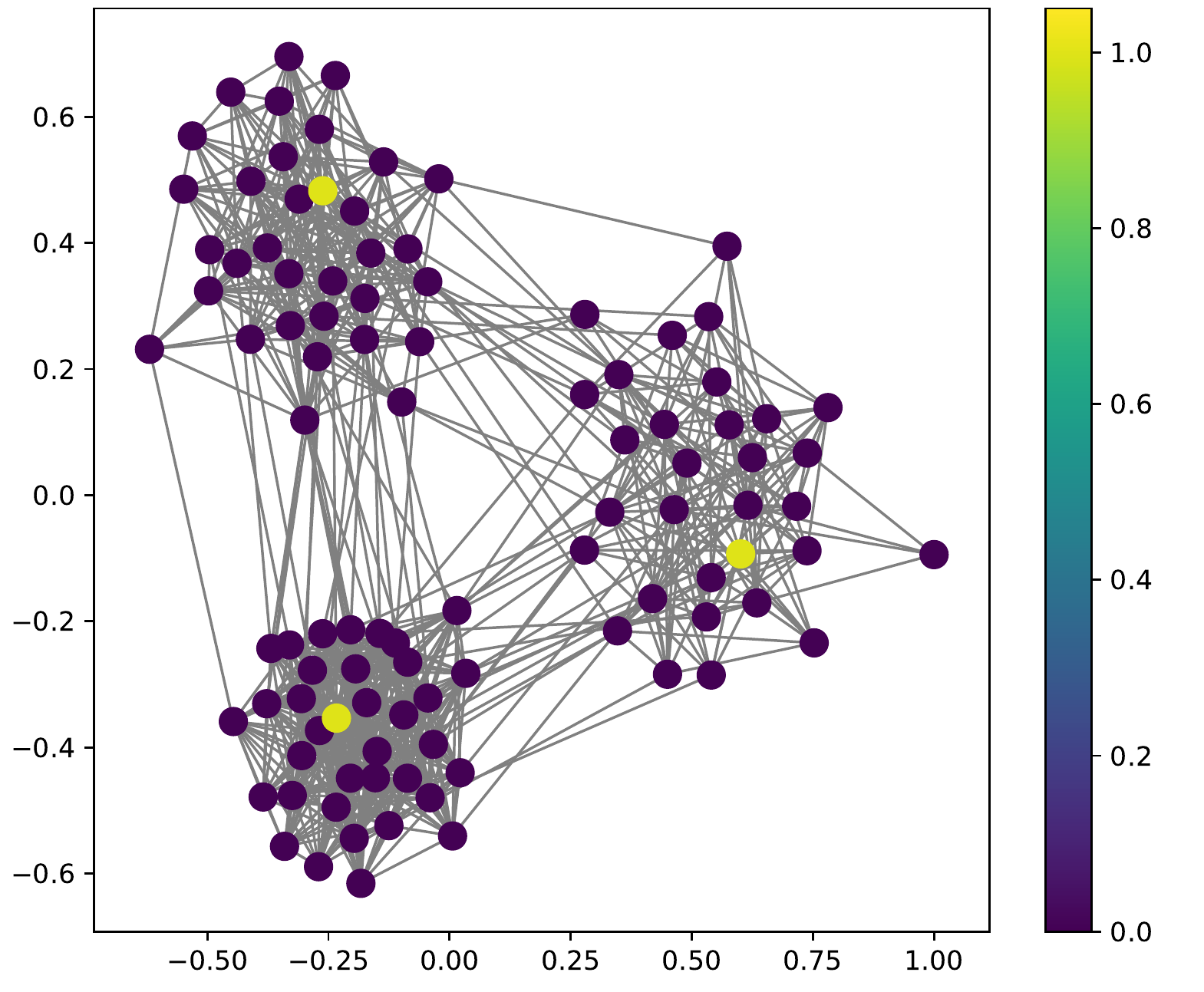}
				\caption{original signal.\protect\linebreak}
				\label{Fig:community_original}
			\end{subfigure}
			\hfill
%            \hspace{-80pt}
			\begin{subfigure}[b]{0.45\hsize}
				\centering
				\includegraphics[width=\hsize]{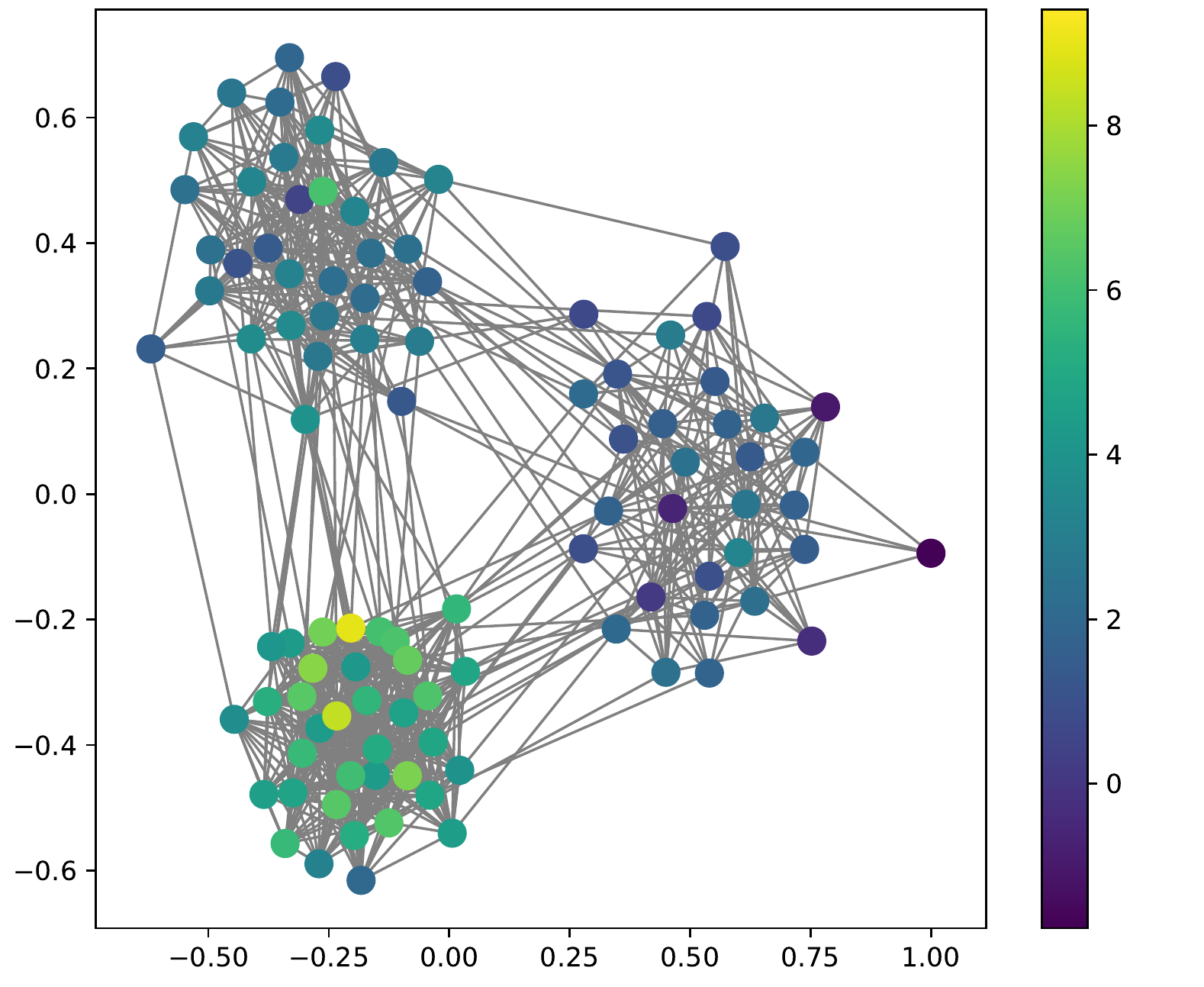}
				\caption{diffused signal (noisy).\protect\linebreak\centering RMSE=$1.4\times 10^{1}$}
				\label{Fig:community_diffusion}
			\end{subfigure}
			\vskip\baselineskip
%             \vspace{-10pt}
			\begin{subfigure}[b]{0.45\hsize}
				\centering
				\includegraphics[width=\hsize]{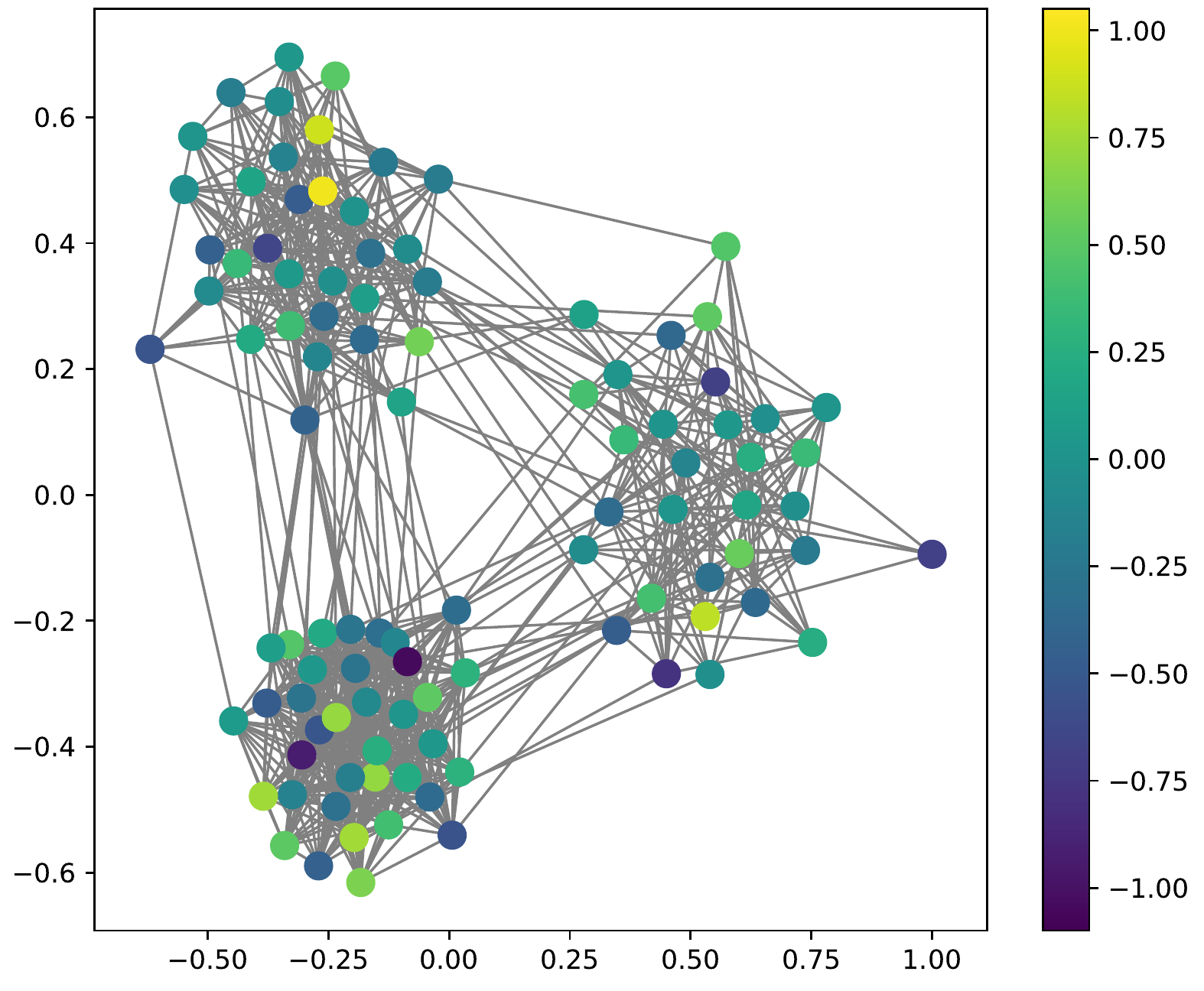}
				\caption{convex relaxation \cite{blind_identification}.\protect\linebreak\centering RMSE=$3.8\times 10^{-1}$}
				\label{Fig:community_convex}
			\end{subfigure}
			\hfill
% 			\hspace{-80pt}
			\begin{subfigure}[b]{0.45\hsize}
				\centering
				\includegraphics[width=\hsize]{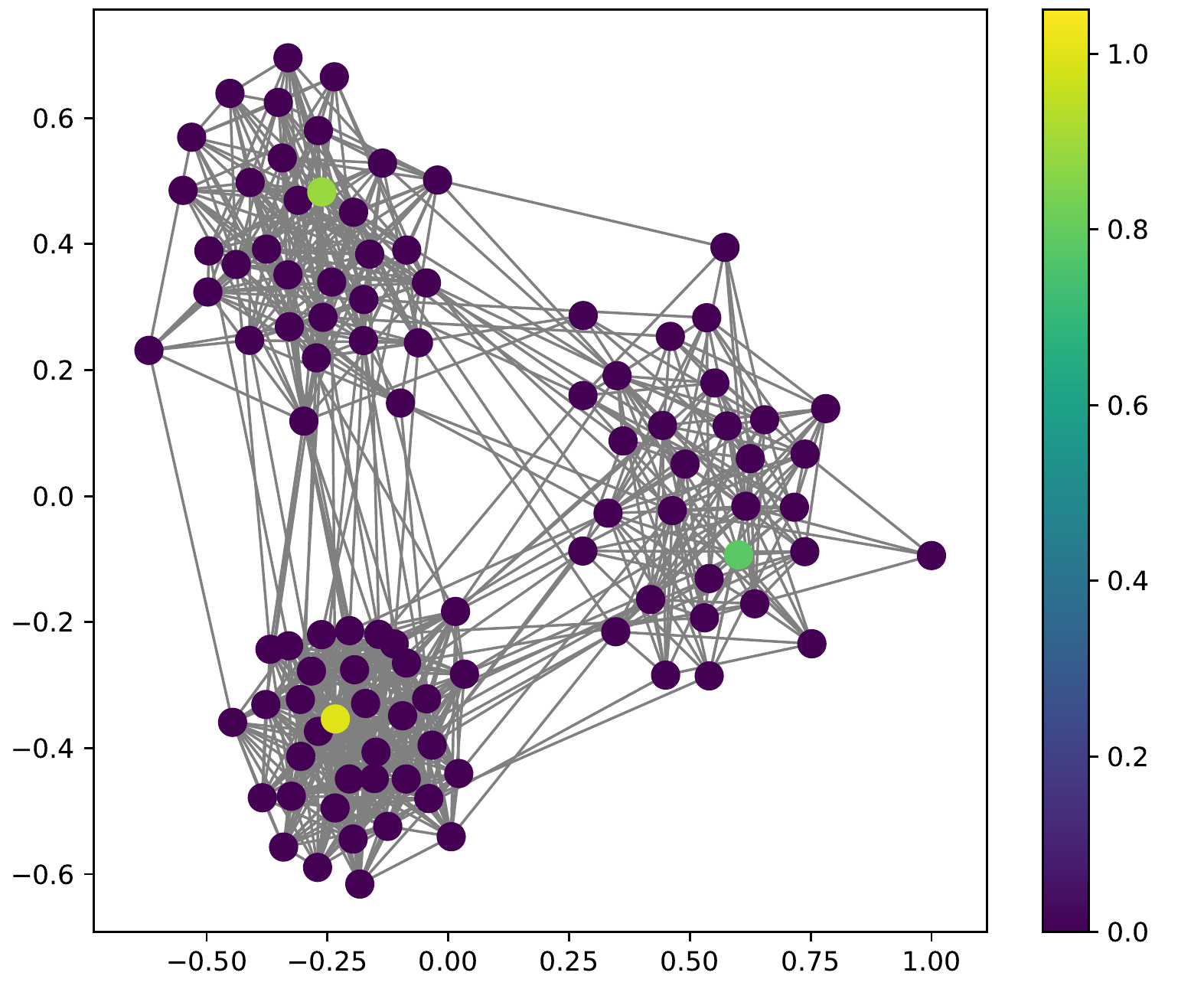}
				\caption{$S$-sparse constraint.\protect\linebreak\centering RMSE=$2.2\times 10^{-2}$}
				\label{Fig:community_ssparse}
			\end{subfigure}
			\caption{Recovery results on a community graph with $N = 100$ and $L = 5$ (mismatched case).}
			\label{Fig:community_graphs}
		\end{figure}
		In many cases, the order of the graph filter is not known a priori when restoring real data. Therefore, in this experiment, we set the order of the graph filter (\ref{Eq:graph_filter}) to $5$, which is larger than the actual filter order $L = 3$.\par
		The RMSEs for the results are summarized in Table \ref{Tab:rmse-B}. As in the previous experiment, the proposed method outperformed the existing method. 
		Fig. \ref{Fig:community_graphs}\subref{Fig:community_convex} and Fig. \ref{Fig:community_graphs}\subref{Fig:community_ssparse} show the results of the convex relaxation (\ref{Eq:Graph_blind_deconvolution_3}) and our $S$-sparse constraint, respectively. Because the order of the estimated filter is different from that of the actual diffusion filter, high magnitude values appear in many places other than the original signal source in Fig. \ref{Fig:community_graphs}\subref{Fig:community_convex}. By contrast, only the signal sources can be accurately estimated by constraining the number of signal sources shown in Fig. \ref{Fig:community_graphs}\subref{Fig:community_ssparse}.
		\begin{table}[]
			\caption{Average RMSE of restored signals in Experiment \ref{sub:experiment-B}}
			\label{Tab:rmse-B}
			\centering
			\setlength{\tabcolsep}{5pt}
	   % 	\small
			\begin{tabularx}{\hsize}{c|ccc}
				\hline
				\thead{Graph}    & \thead{Diffused signal} & \thead{Convex relaxation} & \thead{Proposed} \\
				\hline
				Sensor   & 3.8 & $3.9\times 10^{-1}$  & $\bm{2.1\times 10^{-1}}$ \\
				Community & $1.5\times 10^{1}$ & $4.5\times 10^{-1}$ & $\bm{3.1\times 10^{-2}}$ \\
				\hline
			\end{tabularx}
		\end{table}
	% subsection community_graph (end)
	\subsection{Recovery performance for different initial values} % (fold)
	\label{sub:recovery_performance}
		Finally, we compared the performances between a recovery using random initial values and a recovery using the optimal solutions of the convex relaxed problem (\ref{Eq:Graph_blind_deconvolution_3}). In this experiment, the location of the signal source are determined at random. The filter coefficients of the graph filter are also set at random within the range of $[0, 1]$. The restoration performance is evaluated based on the ratio of the restored signal to the original signal, as given by the following: 
		\begin{align}
			r_{\mathrm{restore}}=\frac{1}{N}\sum_{k=1}^N\frac{S_{k,\mathrm{match}}}{S_k},
		\end{align}
		where $N$ is the number of trials and $S$ and $S_{k,\mathrm{match}}$ represent the number of original sources and the number of matched sources, respectively. Fig. \ref{Fig:recovery_performance} shows the recovery performance for each combination of the number of original sources and the order of the graph filter for signals on the random sensor graph. Because random initial values tend to fall into the local minima, the recovery performance is low even when the number of original sources is small or the order of the filter is low. When the optimal solution of the convex problem is set to the initial value, it is observed that the recovery performance is considerably improved. 
		\begin{figure}
        	\centering
			\begin{subfigure}[b]{0.49\hsize}
				\centering
				\includegraphics[width=\hsize]{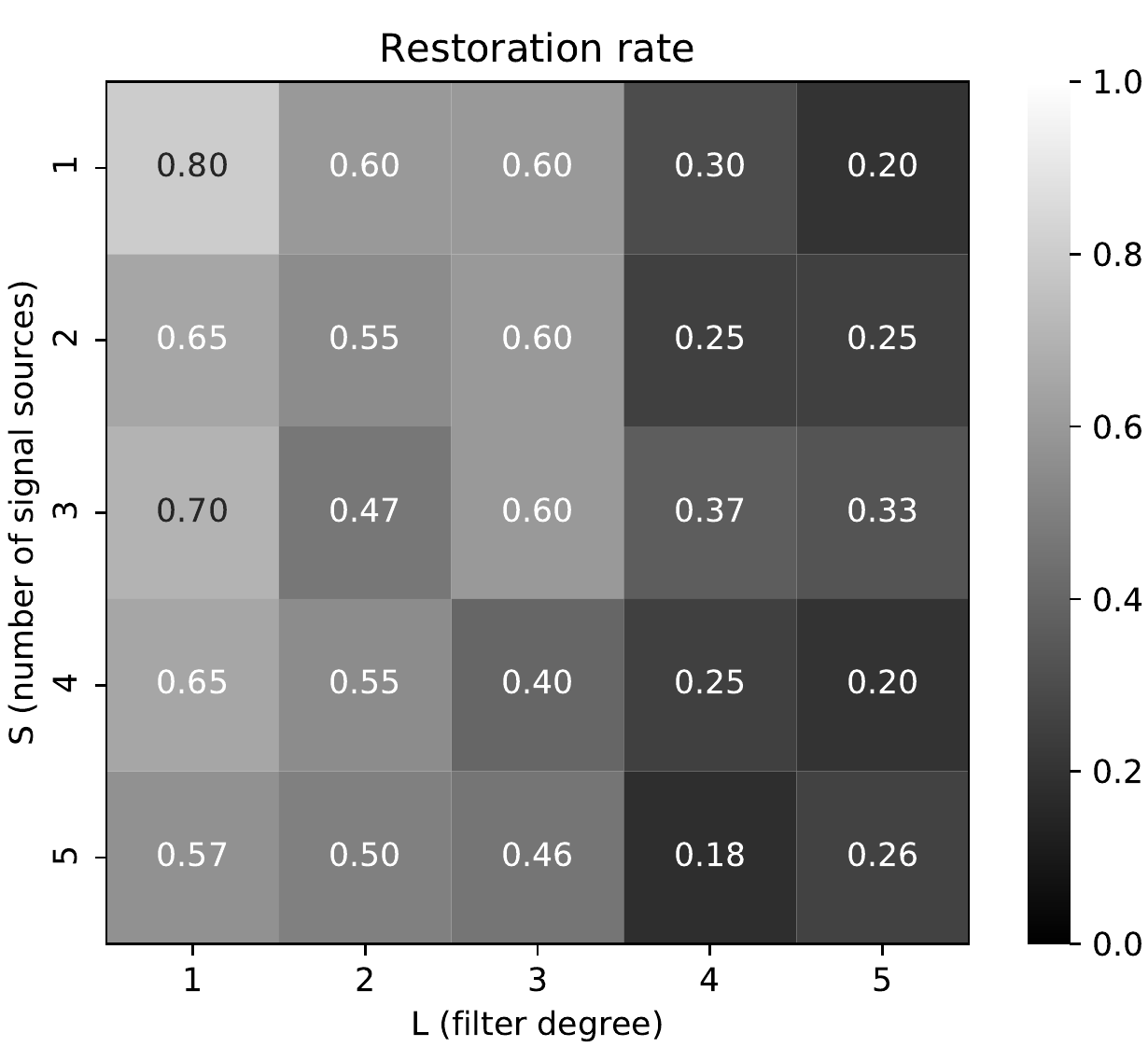}
				\caption{Random initial value.}
				\label{Fig:random}
			\end{subfigure}
			\hfill
% 			\hspace{-60pt}
			\begin{subfigure}[b]{0.49\hsize}
				\centering
				\includegraphics[width=\hsize]{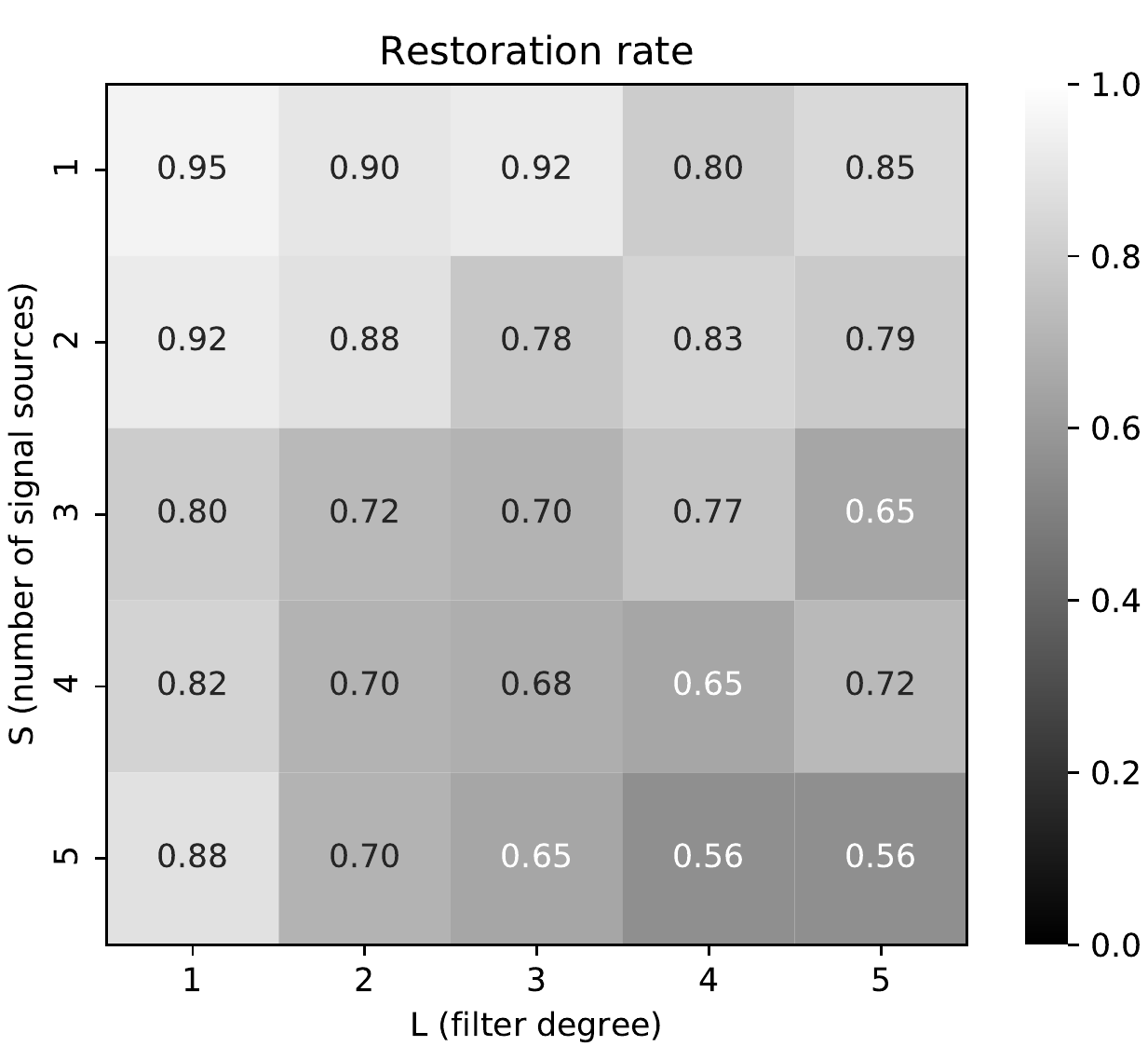}
				\caption{Optimal solution of (\ref{Eq:Graph_blind_deconvolution_3}).}
				\label{Fig:solution_of_convex_relaxed}
			\end{subfigure}
			\caption{Recovery performance matrix according to $S$ and $L$.}
			\label{Fig:recovery_performance}
		\end{figure}
	% subsection recovery_performance (end)
\section{Conclusion}
	We propose a method for identifying the original graph signal from diffused noisy measurements with the exact sparseness constraint. Our ADMM-based algorithm can recover the original signal based on a non-convex optimization problem with the $\ell_0$ constraint. Numerical results demonstrate the superiority of the proposed approach over existing methods. Furthermore, we showed that appropriately setting the initial values improves the restoration performance. 

% Generated by IEEEtran.bst, version: 1.14 (2015/08/26)

\end{document}